\documentclass[letterpaper, 10 pt, conference]{ieeeconf}  

\IEEEoverridecommandlockouts                              

\usepackage{graphicx} 
\usepackage{amsmath} 
\usepackage{amssymb}  
\usepackage{hyperref}
\usepackage{soul}
\usepackage{float}
\usepackage{bm}
\usepackage{cite}
\usepackage{booktabs}

  \usepackage{amsthm}
\theoremstyle{definition}

\theoremstyle{definition}

\theoremstyle{plain}

\theoremstyle{plain}
\newtheorem{propo}{Proposition}
\theoremstyle{definition}
\theoremstyle{definition}

\theoremstyle{definition}

\newtheorem{remark}{Remark}

\usepackage[table,xcdraw]{xcolor}
\usepackage{easyReview}

\title{\LARGE \bf
Robust Variable-Horizon MPC for Landing a Multirotor UAV on a Moving Platform
}

\author{}
\author{Sander Doodeman$^{1}$, Niels Berkers$^{1}$, Paula Chanfreut$^{1}$, Elena Torta$^{1}$, and Duarte Antunes$^{1}$
\thanks{$^{1}$ Department of Mechanical Engineering, TU/e, Eindhoven University of Technology, Eindhoven, The Netherlands}%
\thanks{Email: {\tt s.doodeman@tue.nl}}%
\thanks{This research is supported by the Renewable Energy Transition Topsector Energy subsidies (HER+22-02-03430912) from the Dutch Ministry of Economic Affairs.}%
}

\begin{document}

\maketitle
\thispagestyle{empty}
\pagestyle{empty}

\begin{abstract}
Landing a multirotor Unmanned Aerial Vehicle (UAV) on a moving platform is challenging because a UAV is underactuated and the desired landing state is generally a non-equilibrium state.
Shrinking- and variable-horizon approaches are promising for reaching such non-equilibrium targets, but often lack robustness to disturbances.
Robust Variable-Horizon Model Predictive Control (VH-MPC) addresses this limitation but is computationally complex for a high-dimensional system such as a multirotor UAV.

This paper presents a computationally efficient robust Variable-Horizon MPC method for reaching non-equilibrium targets under bounded disturbances.
The online horizon search is restricted to a neighborhood of the previously selected horizon, while maintaining recursive feasibility under bounded disturbances.
A fixed robust positively invariant tube provides horizon-independent constraint tightening.
By exploiting the differential flatness property of a UAV, this approach is applied to a decoupled linearized system, enabling real-time implementation on an onboard Raspberry Pi 5 at 20 Hz.
Simulations show that restricting the horizon search substantially reduces computation with limited impact on the objective, while real-time Gazebo simulations demonstrate the robustness of the method.
Physical experiments demonstrate successful landing of a multirotor UAV on a moving Stewart platform.

\end{abstract}

\section*{Supplementary Material}
\noindent\textbf{Video:} \href{https://youtu.be/qeYP8OFipIc}{\texttt{youtu.be/qeYP8OFipIc}} \\
\noindent\textbf{Repository:} \href{https://gitlab.tue.nl/tue-waari/rvhmpc}{\texttt{gitlab.tue.nl/tue-waari/rvhmpc}}

\section{INTRODUCTION} \label{sec:intro}
The application of autonomous UAVs has the potential to transform offshore infrastructure inspection.
Assets such as wind turbines, ships, and oil and gas platforms are still inspected mostly by manual inspection, at high cost and in hazardous conditions \cite{fun_sang_cepeda_exploring_2023}.
UAV-based methods, including aerial radiography \cite{meere_x-ray_2025}, can perform the same work faster and without exposing people to that hazard.
However, the further adoption of these platforms is severely restricted by a limited time-of-flight, which requires time-optimal planning and fast mission execution.
Enabling UAVs to safely land on a nearby vessel for recharging during extensive inspections greatly increases the flexibility in task execution.
Having a fast and reliable way to land during a variety of weather conditions is therefore essential, especially when the battery is running low, leaving little time for the landing maneuver.

However, landing successfully on a moving vessel is generally difficult, especially under harsh sea conditions, when the vessel experiences pronounced roll and pitch motions.
A standard multirotor UAV has four independent control inputs (total thrust and three rotation moments), but six degrees of freedom (DOFs) and is therefore underactuated.
Therefore, the 6-DOF landing platform on a vessel is a non-equilibrium target, creating a challenging control problem.
A trajectory tracking controller is therefore not a feasible solution.
The presence of disturbances, such as wind disturbances and the ground effect \cite{sanchez-cuevas_characterization_2017, kan_analysis_2019}, makes landing on a moving vessel an even more complicated task.

\begin{figure}[t]
    \centering
    \includegraphics[width=0.77\linewidth]{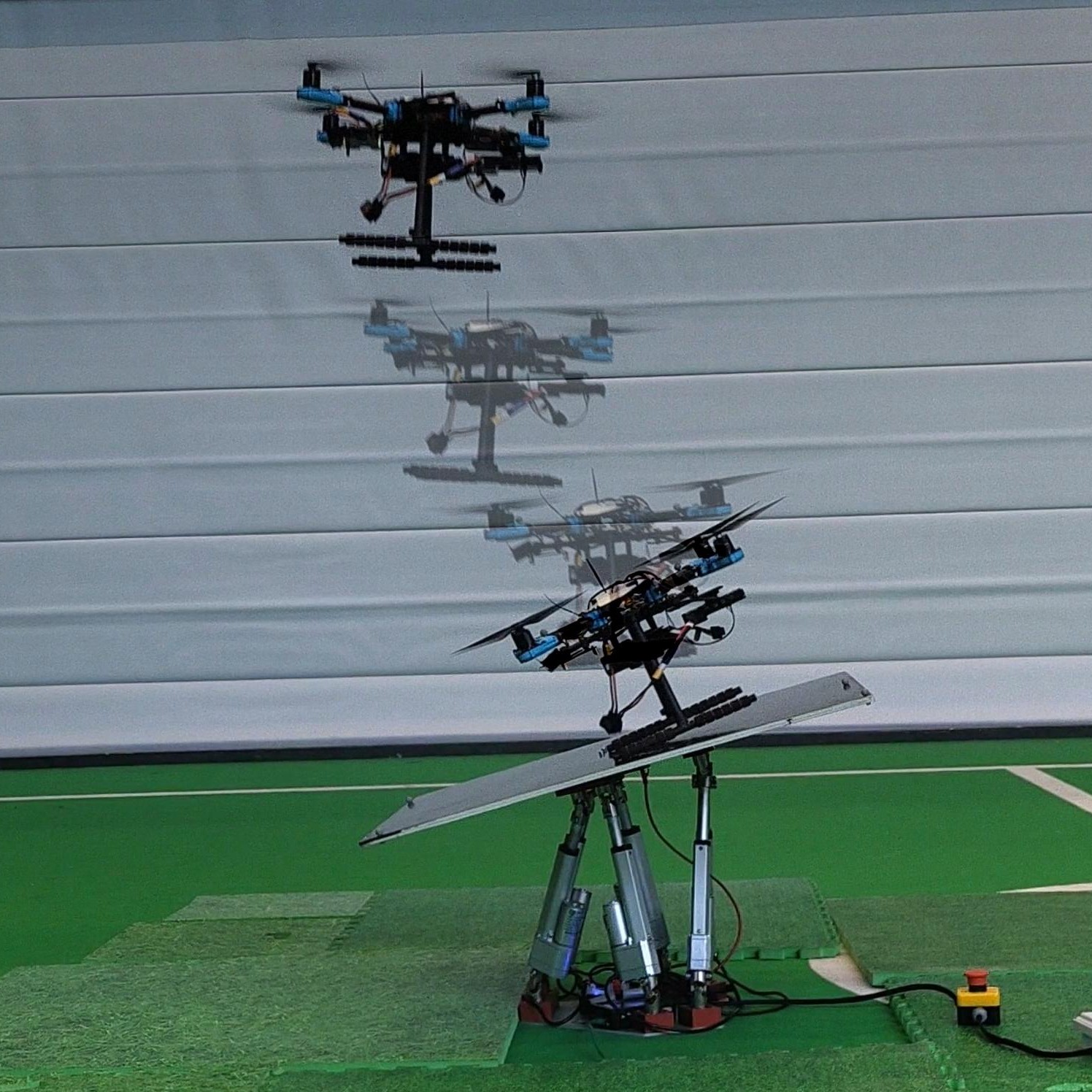}
    \caption{Experimental setup where a Holybro X500 successfully lands on a 6-DOF Acrome Stewart Stage mimicking vessel motion}
    \vspace{-10pt}
    \label{fig:experimental_setup}
\end{figure}

One solution is to equip the UAV with tilting rotors, allowing 6-DOF actuation \cite{lee_autonomous_2024, iida_adaptive_2025}, or to design a custom actuated landing platform to cancel the motion of the vessel \cite{liu_research_2024}.
However, these solutions require additional hardware and a custom UAV with an added payload for actuation, or rely on a communication link to the vessel.
Therefore, we focus on solutions within the UAV control framework.

Existing solutions include a dynamic cost Model Predictive Control (MPC) method \cite{prochazka_model_2024}, which starts prioritizing orientation of the UAV over tracking the position as it gets closer to the moving vessel.
Similarly, \cite{stephenson_distributed_2024} and \cite{gupta_landing_2023} use an MPC position tracking controller to land on a moving vessel, while minimizing tilt during landing.
A disadvantage of these position tracking methods is that they always introduce an error when reaching the non-equilibrium point, and an appropriate descent rate has to be picked to make sure the UAV is not too low before reaching the right horizontal position.
Greer and Sultan propose a shrinking horizon MPC approach to land a helicopter on a moving platform, allowing it to reach a non-equilibrium point \cite{greer_shrinking_2020}.
However, this method is only verified in simulation.
In \cite{vlantis_quadrotor_2015}, landing on a moving mobile robot with an inclined surface is achieved by introducing a terminal set at the platform within an MPC formulation and optimizing the horizon length.
However, all methods in \cite{prochazka_model_2024, greer_shrinking_2020, vlantis_quadrotor_2015} do not consider disturbances, making them vulnerable, for example, to wind forces acting on the UAV.
This is especially troublesome as recursive feasibility is then not assured; disturbances may drive the system into regions where the MPC optimization problem is no longer guaranteed to remain feasible.
This results in the undesired need to have a fallback scenario (e.g., switching to a simpler controller) to provide the control input, and the properties promised by MPC are lost.

Deep reinforcement learning methods have been applied as well \cite{rodriguez-ramos_deep_2019, habas_ceilings_2025}.
Although these methods are promising, they suffer from a lack of robustness guarantees, which is essential when the UAVs carry expensive inspection equipment, such as cameras or X-ray devices, e.g., for offshore asset inspections \cite{meere_x-ray_2025, zhang_inspection_2024}.
As multirotor UAVs are underactuated, matching the position and orientation of a vessel requires optimizing the trajectory of the UAV over a given time, rather than tracking a state.
VH-MPC does this and also optimizes the arrival horizon online.
However, in the presence of disturbances, VH-MPC loses recursive feasibility.
To guarantee robustness against bounded disturbances, \cite{richards_robust_2006} combines VH-MPC with tube MPC, also known as robust MPC, providing recursive feasibility.
The tube MPC formulation introduces additional computational complexity through the horizon-dependent propagation of disturbance sets and the resulting tightened constraints.
This can be avoided by employing a single robust positively invariant (RPI) set as the tube cross-section, yielding horizon-independent constraint tightening \cite{mayne_robust_2005}.

The optimization over the full horizon length in VH-MPC can result in high computational costs, as multiple candidate horizons must be considered online \cite{richards_robust_2006}.
This is limiting for a multirotor UAV, where a body rate controller must run at a high rate to match the tilt of the UAV with the moving platform while landing.
Moreover, such a high control rate requires small step sizes, requiring a longer horizon length with longer computation times to still include the full landing procedure.

To the best of the authors' knowledge, no robust variable-horizon approach exists that is efficient enough to apply to a real-time multirotor UAV landing on a 6-DOF moving platform in the presence of bounded disturbances.
This paper closes that gap for a 6-DOF vessel.
The contributions of this work are threefold.
First, we develop a computationally efficient Robust Variable-Horizon MPC formulation for robustly reaching a non-equilibrium target under bounded disturbances.
The computational cost of the proposed variable-horizon approach is reduced by restricting the online horizon search to a neighborhood of the previously selected horizon.
In addition, using a fixed RPI tube makes the constraint tightening independent of the prediction horizon, thereby simplifying the optimization while retaining recursive feasibility guarantees.
Second, we develop a practical control architecture that applies a robust VH-MPC formulation to the landing of a UAV on a 6-DOF moving platform, which integrates an outer-loop horizon selection procedure for a decoupled linearized UAV model.
Differential flatness allows decoupling of the linearized system \cite{mellinger_minimum_2011}, which makes the calculation of the fixed RPI tube feasible, and it also allows solving the robust VH-MPC control inputs under a high control rate.
Third, we experimentally demonstrate the effectiveness of the resulting controller for landing maneuvers.
The experiments validate both the real-time computational performance and the ability to robustly reach a non-equilibrium landing state.

This paper is organized as follows: Section \ref{sec:preliminaries} provides the preliminaries.
In Section \ref{sec:rvhmpc}, we present the robust VH-MPC method to reach a non-equilibrium point under bounded disturbances.
Section \ref{sec:uav_adaptation} elaborates on the framework that applies the robust VH-MPC method to a multirotor UAV to land on a moving platform in real-time.
The experiments that were performed to validate the introduced framework are described in Section \ref{sec:experiments}.
Finally, Section \ref{sec:conclusion} provides concluding remarks.

\section{PRELIMINARIES} \label{sec:preliminaries}
Let us denote the Minkowski sum and Minkowski (Pontryagin) set difference as $\oplus$ and $\ominus$, respectively.
Specifically, given two sets $\mathcal{A}$ and $\mathcal{B}$, $\mathcal{A} \oplus \mathcal{B} = \{ \bm{a} + \bm{b} : \bm{a} \in \mathcal{A}, \bm{b} \in \mathcal{B} \}$ and $\mathcal{A} \ominus \mathcal{B} = \{ \bm{a} : \{\bm{a}\} \oplus \mathcal{B} \subseteq \mathcal{A} \}$.

Throughout this paper, we consider the discrete linear time-invariant system
\begin{equation} \label{eq:real_system}
    \bm{x}(k+1) = A \bm{x}(k) + B \bm{u}(k) + \bm{w}(k),
\end{equation}
where $\bm{x}(k) \in \mathbb{R}^n$, $\bm{u}(k) \in \mathbb{R}^m$, and $\bm{w}(k) \in \mathbb{R}^n$ are the system state, control input, and disturbance at time $k\in\mathbb{N}_+$, respectively.
System matrices $A$ and $B$ define the linear dynamics of the system.
The state and input constraints are described by the convex sets $\mathcal{X}\subseteq\mathbb{R}^{n}$ and $\mathcal{U}\subseteq\mathbb{R}^{m}$, respectively.
The disturbance is assumed to be bounded by the compact convex set $\mathcal{W}\subseteq\mathbb{R}^{n}$, where $0 \in \mathcal{W}$.
Within the tube MPC formulation \cite{mayne_robust_2005}, the trajectory of the nominal system is optimized, and the implemented control action is derived from the first element of the resulting nominal input sequence.
Specifically, the nominal system is given by
\begin{equation} \label{eq:nominal_system}
    \bm{z}(k+1) = A \bm{z}(k) + B \bm{v}(k),
\end{equation}
where $\bm{z}(k) \in \mathbb{R}^n$ and $\bm{v}(k) \in \mathbb{R}^m$ are the nominal states and inputs, respectively.
We assume that the system is controllable and thus a linear feedback gain $K$ can be derived such that $(A + BK)$ is Schur.
Then, the control input for the real system \eqref{eq:real_system} is given by
\begin{equation}
    \bm{u}(k) = \bm{v}(k) + K(\bm{x}(k) - \bm{z}(k)) = \bm{v}(k) + K\bm{e}(k),
\end{equation}
where $\bm{e}(k)=\bm{x}(k)-\bm{z}(k)$ denotes the error between the real and nominal states.
The error dynamics evolve as
\begin{equation}
    \bm{e}(k+1) = (A + BK) \bm{e}(k) + \bm{w}(k).
\end{equation}
Then, as proved in \cite{mayne_robust_2005} under the underlying assumptions mentioned above, an RPI set $\mathcal{E}$ exists such that $(A+BK)\mathcal{E}\oplus\mathcal{W}\subseteq\mathcal{E}$.
Furthermore, we assume that $\mathcal{X} \ominus \mathcal{E}$ and $\mathcal{U} \ominus K\mathcal{E}$ are nonempty.

\section{ROBUST VARIABLE-HORIZON MPC} \label{sec:rvhmpc}
\begin{figure*}
    \centering
    \vspace{5pt}
    \includegraphics[width=0.8\linewidth]{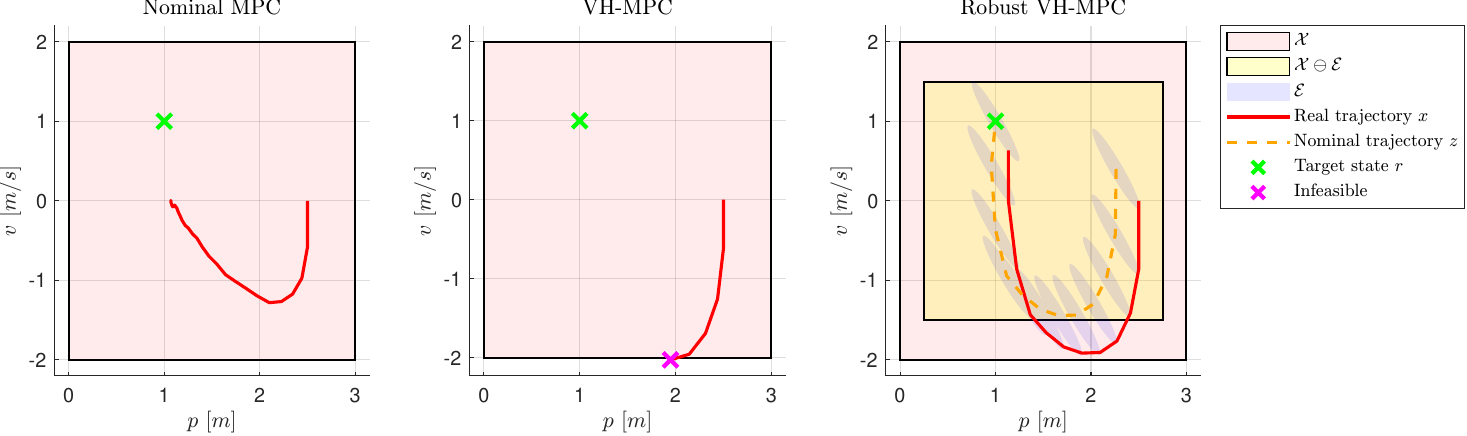}
    \caption{Trajectories for a simple double integrator system from initial state $[2.5\,\,\,\, 0]^\top$ to non-equilibrium target state $[1\,\,\,\, 1]^\top$ (green cross) obtained with robust VH-MPC, VH-MPC, and nominal MPC approaches under bounded disturbances.
    The light red boxes indicate the constraint set $\mathcal{X}$, the yellow box shows the shrunken set $\mathcal{X}\ominus\mathcal{E}$, and the light blue regions represent tube $\mathcal{E}$.
    Finally, the magenta cross shows an infeasible state due to state or input bound constraints. See the full implementation at \href{https://gitlab.tue.nl/tue-waari/rvhmpc}{\texttt{gitlab.tue.nl/tue-waari/rvhmpc}}}
    \vspace{-5pt}
    \label{fig:vh_explanation}
\end{figure*}
This section presents a computationally efficient robust VH-MPC method that enables robustly reaching a non-equilibrium point, using an efficient horizon search.

Using robust VH-MPC is motivated by the fact that nominal MPC fails to reach a non-equilibrium target, as Fig. \ref{fig:vh_explanation} illustrates.
For an objective function that, over the full horizon length, penalizes the error in position and velocity, as well as the input, the nominal MPC cannot attain the target state.
VH-MPC could reach the target in the absence of disturbances but becomes infeasible since it is pushed outside the constraint set $\mathcal{X}$ due to the given bounded disturbances.
Robust VH-MPC combines VH-MPC with tube MPC to guarantee recursive feasibility, making it a suitable choice for landing a UAV on a moving platform under bounded disturbances.
This is why we base our formulation on \cite{richards_robust_2006}.

At each robust VH-MPC update instant $k$, the predicted nominal state and input sequences are denoted by $Z_k = [\bm{z}_k(0), \bm{z}_k(1), ..., \bm{z}_k(N_k)]$ and $V_k = [\bm{v}_k(0), \bm{v}_k(1), ..., \bm{v}_k(N_k-1)]$, respectively, with prediction horizon $N_k$.
Given the target state $\bm{r}(k) \in \mathbb{R}^n$ at step $k$, we define the following nominal cost function, similar to \cite{richards_robust_2006}:
\begin{gather} \label{eq:cost_fcn}
\begin{aligned}
    J(Z_k, V_k, N_k) = &\sum_{j=0}^{N_k}||\bm{z}_k(j) - \bm{r}(k+j)||^2_{Q} \\ &+ \sum_{j=0}^{N_k-1} ||\bm{v}_k(j)||^2_{R} + \gamma N_k.
\end{aligned}
\end{gather}
Here, positive definite $Q$ and $R$ are state and input cost matrices, and $\gamma\geq0$ is a weighting factor that enforces a reduction of the horizon length.
Cost \eqref{eq:cost_fcn} is to be minimized for the horizon length $N_k$, nominal state sequence sequence $Z_k$ and control input sequence $V_k$ via the following robust VH-MPC problem:
\begin{subequations}
\label{eq:rvhmpc_problem}
\begin{align}
    J^\ast = & \min_{N_k} \min_{Z_k,V_k} J(Z_k, V_k, N_k), \label{eq:ocp} \\
    \text{s.t.} \,\, & \bm{x}(k) \in \{\bm{z}_k(0)\} \oplus \mathcal{E}, \label{eq:ocp1} \\
    & \bm{z}_k(j+1) = A \bm{z}_k(j) + B \bm{v}_k(j), \quad j\in\{0, ..., N_k-1\}, \\
    & \bm{z}_k(j) \in \mathcal{X} \ominus \mathcal{E} \text{ for } j\in\{0, ..., N_k-1\}, \\
    & \bm{v}_k(j) \in \mathcal{U} \ominus K \mathcal{E} \text{ for } j\in\{0, ..., N_k-1\}, \\
    & \bm{z}_k(N_k) = \bm{r}(k+N_k), \label{eq:ocp5}\\
    & N_k \in
    \left[
    N_{k-1}^\ast - 1 - \Delta N,\,
    N_{k-1}^\ast - 1 + \Delta N
    \right] \cap \mathbb{N}_+, \label{eq:n_in_new}
\end{align}
\end{subequations}
where $\Delta N$ defines the search range of horizon lengths around the optimal horizon length $N_{k-1}^\ast$ from the previous iteration.
Based on the optimal horizon length $N_k^\ast$, the applied nominal control input is then picked as the first control input in 
\begin{equation}
    V_k^\ast = \mathop{\arg \min} \limits_{V_k^{N_k^\ast}} J(V_k^{N_k^\ast}, N_k^\ast).
\end{equation}

Since solving \eqref{eq:rvhmpc_problem} is computationally demanding for a large set of horizon lengths, we modify the formulation with respect to \cite{richards_robust_2006} by only considering a range of horizons around the previous optimal horizon length in \eqref{eq:n_in_new}.
We now prove that recursive feasibility is still maintained for this formulation.

\begin{propo} \label{propo:recursive_feasibility}
    If problem \eqref{eq:rvhmpc_problem} is feasible at time \(k\) with optimal horizon \(N_k^\ast\geq 1\), then it is feasible at time \(k+1\) for horizon \(N_{k+1}=N_k^\ast-1\).
\end{propo}

\begin{proof}
Let \(\{\bm{z}_k^\ast(j),\bm{v}_k^\ast(j)\}\) denote a feasible solution at time \(k\).
Since
\begin{equation*}
\bm{x}(k) \in \{\bm{z}_k^\ast(0)\} \oplus \mathcal E,
\end{equation*}
the tube error satisfies \(\bm{e}(k)\in\mathcal E\).
By the RPI property,
\begin{equation*}
\bm{e}(k+j)\in\mathcal E, \qquad \forall j\geq 0,
\end{equation*}
and therefore
\begin{equation*}
\bm{x}(k+j) \in \{\bm{z}_k^\ast(j)\} \oplus \mathcal E, \qquad \forall j\geq 0.
\end{equation*}
In particular,
\begin{equation*}
\bm{x}(k+1) \in \{\bm{z}_k^\ast(1)\} \oplus \mathcal E.
\end{equation*}
Hence, the shifted sequence
\begin{equation*}
\hat{\bm{z}}_{k+1}(j)=\bm{z}_k^\ast(j+1),\qquad
\hat{\bm{v}}_{k+1}(j)=\bm{v}_k^\ast(j+1)
\end{equation*}
satisfies the initial, state, and input constraints at time \(k+1\).
Moreover,
\begin{align*}
\hat{\bm{z}}_{k+1}(N_k^\ast-1)
=\bm{z}_k^\ast(N_k^\ast)
&=r(k+N_k^\ast) \\
&=r\!\left((k+1)+(N_k^\ast-1)\right),
\end{align*}
and thus satisfies the terminal equality.
Consequently, the OCP is feasible at time \(k+1\) with horizon \(N_k^\ast-1\).
Since \(N_k^\ast-1\) belongs to the feasible horizon set in \eqref{eq:n_in_new} for any \(\Delta N\geq0\), recursive feasibility follows.
\end{proof}

Fig. \ref{fig:rvhmpc_cost_map} illustrates the influence of $\gamma$ on the task execution, where a higher value of $\gamma$ enforces a faster time of arrival.
A low value of $\gamma$ will result in a smoother trajectory (by tuning $Q$) and/or a smoother input profile (by tuning $R$).
\begin{figure}[t]
    \centering
    \vspace{5pt}
    \includegraphics[width=0.95\linewidth]{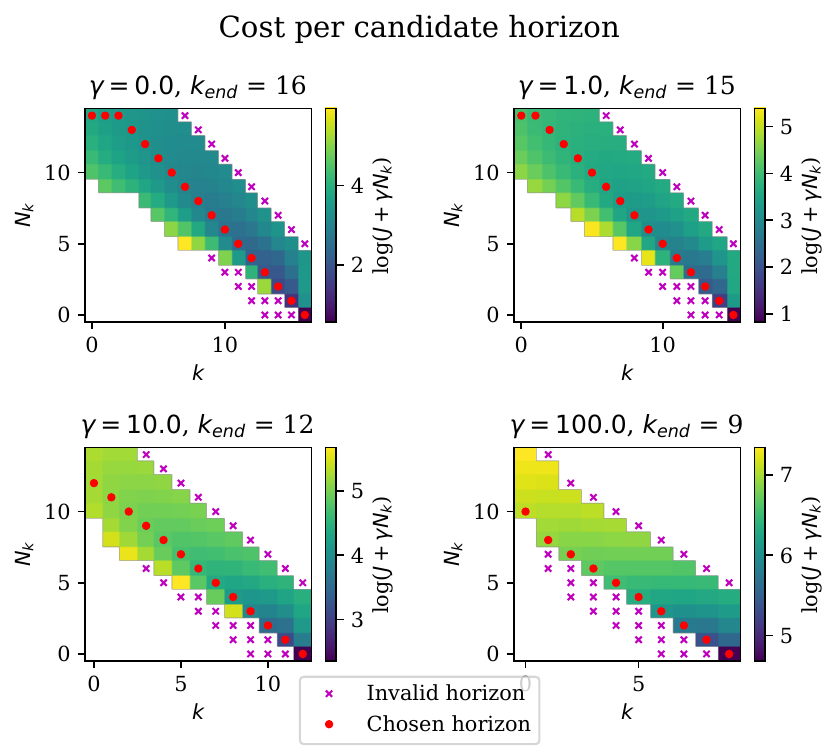}
    \vspace{-5pt}
    \caption{Costs for the evaluated horizons for each timestep for a simulation of a simple 3D system.
    The red dots indicate the chosen $N_k$, the magenta crosses indicate infeasible $N_k$.
    A constant reference state $r_k$ is used for these simulations.
    See our repository for details.}
    \label{fig:rvhmpc_cost_map}
\end{figure}

Since this controller is designed to run at a high frequency, depending on the chosen value of $\Delta N$, the horizon can quickly adjust to an optimal value over time.
This is also visible in Fig. \ref{fig:rvhmpc_cost_map}, where for high values of $\gamma$, the lowest horizon length in the given range in \eqref{eq:n_in_new} is only selected once.
Note that for $\Delta N=0$, the approach is identical to a standard shrinking horizon approach.

\section{ROBUST VH-MPC-BASED UAV CONTROL ARCHITECTURE} \label{sec:uav_adaptation}
To apply robust VH-MPC at a high control rate for UAV landing, we combine the formulation of Section \ref{sec:rvhmpc} with a decoupled linearized UAV model and a fixed RPI tube.

\begin{figure}[t]
    \centering
    \includegraphics[width=0.99\linewidth]{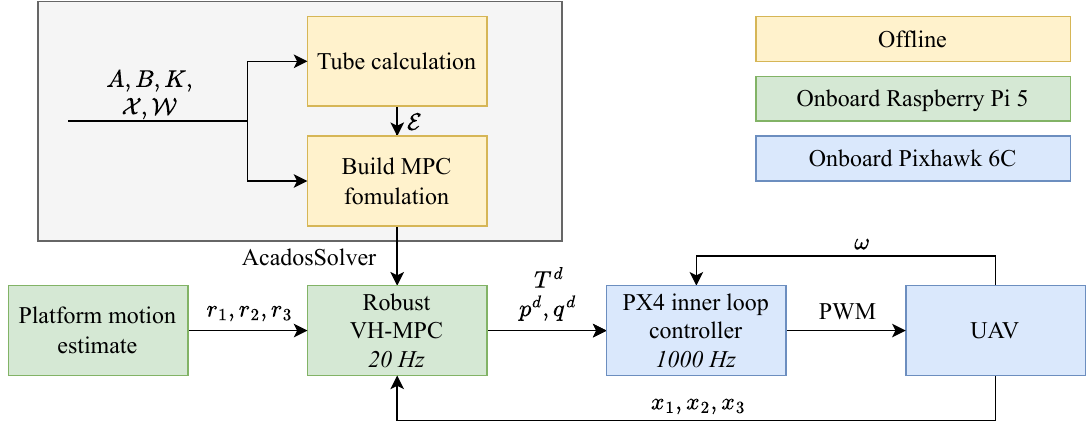}
    \vspace{-10pt}
    \caption{Robust VH-MPC control loop for the multirotor UAV}
    \vspace{-5pt}
    \label{fig:control_loop}
\end{figure}
\begin{figure}[t]
    \centering
    \includegraphics[width=0.8\linewidth]{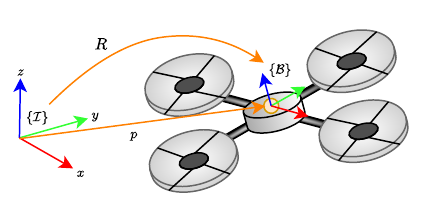}
    \vspace{-10pt}
    \caption{Used UAV coordinate frame}
    \vspace{-10pt}
    \label{fig:coordinates}
\end{figure}
Fig. \ref{fig:control_loop} shows an overview of the applied control architecture.
For the calculation of the thrust and body rate inputs of the UAV, we apply the robust VH-MPC method from Section \ref{sec:rvhmpc} to a decoupled, linearized model of a multirotor UAV and introduce an inner and outer optimization step to cope with the decoupling of the system.
In the architecture in Fig. \ref{fig:control_loop}, the robust VH-MPC control loop is designed to run at 20 Hz; this computational performance is evaluated experimentally in Section \ref{sec:calc_times}.
The thrust and body rates are maintained in a high-rate inner loop, resulting in the direct PWM motor signals.
Since this inner loop runs at a higher rate, we treat the rate tracking as instantaneous.
For the robust VH-MPC implementation, an RPI set is required.
Section \ref{sec:tube} describes how this RPI set is derived.
We now elaborate on the control design choices that were made to achieve a practical real-time robust VH-MPC control loop for a UAV, starting with the modeling of the system.

\subsection{Modeling and OCP for the multirotor UAV}
The states of the multirotor UAV that are considered relevant for landing on a moving platform are the position $\bm{p}\in\mathbb{R}^3$, velocity $\bm{v}\in\mathbb{R}^3$, and attitude $\bm{\eta}\in\mathbb{R}^3$, for which the latter is decomposed into the roll $\varphi$, pitch $\theta$, and yaw $\psi$ of the UAV.
These angles form the rotation matrix $R \in SO(3)$ that maps any vector in the UAV body-fixed frame $\mathcal{B}$ to the inertial frame $\mathcal{I}$ using the Euler $XYZ$ convention (Fig. \ref{fig:coordinates}).
The considered control inputs for the system are given by the total thrust $T$, and the body rates $p$, $q$, and $r$.
Since the translational dynamics are independent of the yaw of a multirotor UAV, the yaw is omitted from the UAV state and is controlled separately with a PD controller.
This leaves eight states to be controlled.
Although an MPC formulation with a linear model can be solved very efficiently with a Quadratic Programming (QP) solver, this high dimensionality still results in computational complexity.
Therefore, we apply the differential flatness result in \cite{mellinger_minimum_2011} and decouple the system into subsystems with at most three states.

Given the general multirotor UAV system dynamics
\begin{subequations}
\begin{align}
    \dot{\bm{p}} &= \bm{v}, \\
    \dot{\bm{v}} &= -\bm{g} + \frac{1}{m}R(\bm{\eta}) \begin{bmatrix}
        0 & 0 & T
    \end{bmatrix}^\top + F_\text{ext}, \\
    \dot{\bm{\eta}} &= \begin{bmatrix}
        p & q & r
    \end{bmatrix}^\top,
\end{align}
\end{subequations}
where $\bm{g}$ is the gravitational acceleration and $F_\text{ext}$ represents the external disturbance forces acting on the system; we then derive the decoupled linearized dynamics for system states $\bm{x}_1 = \begin{bmatrix}
    p_x & v_x & \theta
\end{bmatrix}^\top$, $\bm{x}_2 = \begin{bmatrix}
    p_y & v_y & \varphi
\end{bmatrix}^\top$, and $\bm{x}_3 = \begin{bmatrix}
    p_z & v_z
\end{bmatrix}^\top$ as
\begin{subequations}\label{eq:linearized_system}
\begin{align}
    \label{eq:sys1}& \dot{\bm{x}}_1 = \begin{bmatrix}
        0 & 1 & 0 \\
        0 & 0 & g \\
        0 & 0 & 0
    \end{bmatrix} \bm{x}_1 + \begin{bmatrix}
        0 \\ 0 \\ 1
    \end{bmatrix} q + \bm{w}_1, \\
    & \dot{\bm{x}}_2 = \begin{bmatrix}
        0 & 1 & 0 \\
        0 & 0 & -g \\
        0 & 0 & 0
    \end{bmatrix} \bm{x}_2 + \begin{bmatrix}
        0 \\ 0 \\ 1
    \end{bmatrix} p + \bm{w}_2, \\
    \label{eq:sys3}& \dot{\bm{x}}_3 = \begin{bmatrix}
        0 & 1 \\
        0 & 0
    \end{bmatrix} \bm{x}_3 + \begin{bmatrix}
        0 \\ 1
    \end{bmatrix} T + \bm{w}_3.
\end{align}
\end{subequations}
This decoupling is central to the real-time implementation: it reduces the RPI-set calculation and MPC subproblems to systems of at most three states, while allowing the horizon selection to be performed over the complete UAV objective, as we show next.

We discretize this linear system and use it within the MPC framework.
To efficiently solve the robust VH-MPC problem in \eqref{eq:rvhmpc_problem}, we first solve the MPC formulation for each subsystem in \eqref{eq:linearized_system} individually for the selected horizon range
\begin{gather}
\begin{aligned}
    J_{i}(\bm{x}_i(k), N_k) = \min_{V_k} &\sum_{j=0}^{N_k}||\bm{z}_{k,i}(j) - \bm{r}_i(k+j)||^2_{Q_i}\\+&\sum_{j=0}^{N_k-1} ||\bm{v}_{k,i}(j)||^2_{R_i}, \quad \text{ s.t. \eqref{eq:ocp1}- \eqref{eq:n_in_new}},
\end{aligned}
\end{gather}
after which we choose the best horizon:
\begin{equation} \label{eq:outer_optimization}
    J^\ast = \min_{N_k} \gamma N_k + \sum_{i\in\{1,2,3\}}J_i(\bm{x}_i(k), N_k),
\end{equation}
where $Q_i$ and $R_i$ are such that $Q = \text{diag}(Q_1, Q_2, Q_3)$ and $R = \text{diag}(R_1, R_2, R_3)$ for $i\in\{1,2,3\}$.


\begin{remark}
    The linearized model is valid around the hovering state of a multirotor UAV.
    Especially for touchdowns with higher angles, this approximation is less accurate.
    However, the discrepancies caused by the nonlinearity of the system can be seen as part of the disturbance $\bm{w}_i$ for $i\in \{1, 2, 3 \}$.
\end{remark}

From the resulting optimal horizon in \eqref{eq:outer_optimization}, we derive the optimal control inputs for each subsystem.

\subsection{Tube calculation} \label{sec:tube}
The robust VH-MPC formulation evaluates multiple candidate horizons at each update.
Using state- or horizon-dependent tube sets would therefore increase the computational burden.
We therefore use a single fixed RPI set $\mathcal{E}$ as the tube cross-section, which can be calculated offline.
The decoupling of the linearized system ensures that such an RPI set only needs to be calculated for up to three dimensions.
However, computing this RPI set is still challenging.

A polytopic representation of tube cross-section $\mathcal{E}$ is adopted in this work, i.e., 
\begin{equation*}
    \mathcal{E} = {\bm{e} \in \mathbb{R}^n : H \bm{e} \leq h},
\end{equation*}
where matrix $H$ and vector $h$ define the half-space representation.
This yields linear constraints and preserves the tractability of problem \eqref{eq:rvhmpc_problem} at high control rates.
However, the number of facets required to represent an exact finite RPI set generally increases rapidly with the state dimension, resulting in a significant computational burden.
To address this and retain the robustness properties, an approximation of the RPI set is used based on \cite{trodden_one-step_2016}.
Under this framework, $\mathcal{E}$ is computed offline by optimizing vector $h$ for a prescribed matrix $H$.
The associated optimization problem is a linear program (LP) and can therefore be solved efficiently.

To obtain a suitable matrix $H$, we first construct an invariant ellipsoid $S(P)=\{\bm{e}\in\mathbb{R}^n:||\bm{e}||^2_P \leq 1\}$ from the closed-loop system matrix $A_{cl} =  A+BK$ and disturbance set $\mathcal{W}$, using \cite{kothare_robust_1996}, such that for system \eqref{eq:real_system}-\eqref{eq:nominal_system}, if $\bm{e} \in S(P)$, then $\bm{e} \in S(P)$ for $i \in \mathbb{N}_+$.
The ellipsoidal invariant set is then approximated by a polytope obtained from supporting hyperplanes.
The hyperplane normals are generated by sampling rotations in each coordinate plane, while the corresponding offsets are computed exactly from the ellipsoid support function $h(\bm{n})=||\bm{n}||_P$.
While this ellipsoidal approximation is not guaranteed to be invariant, it is used only to define the facet normals $H$.
The optimization of offset vector $h$ according to \cite{trodden_one-step_2016} guarantees that the resulting set is RPI.

\begin{remark}
    Numerous methods utilize a dynamic shape of the tube $\mathcal{E}$ for less conservatism and better control performance \cite{gonzalez_online_2011, saccani_homothetic_2023}.
    Although especially at the final stages of landing, high control performance is essential, a dynamic tube comes at a high computational cost.
    Therefore, we have chosen a fixed RPI set, trading potential reduction in conservatism for a reduction in computational complexity, which is important for the considered high-rate implementation.
\end{remark}

A full implementation of the described method can be found at \href{https://gitlab.tue.nl/tue-waari/rvhmpc}{\texttt{gitlab.tue.nl/tue-waari/rvhmpc}}.


\section{EXPERIMENTAL VALIDATION} \label{sec:experiments}
This section validates four aspects required for the proposed approach: the computational feasibility of the proposed controller, its robustness, and the landing performance of a real multirotor UAV on a moving platform.
The range of initial conditions for which the landing problem is feasible is analyzed as well.
Throughout this section, the parameters in Table \ref{tab:used-parameters} are used.
Furthermore, $K$ is chosen to be the LQR feedback gain.
\begin{table}[b]
\centering
\caption{Used parameters throughout simulations and experiments}
\label{tab:used-parameters}
\begin{tabular}{lclc}
\toprule
{Parameter} & {Value} & {Parameter} & {Value} \\
\midrule
Timestep                        & $0.05$ s                          & $Q_1$, $Q_2$ & $\texttt{diag}(2.0,\ 2.0,\ 0.8)$  \\
Maximum horizon                 & $30$                              & $Q_3$ & $\texttt{diag}(30.0,\ 10.0)$  \\
$\Delta N$                      & $1$                               & $R_1$, $R_2$ & $6.0$       \\
$\gamma$                        & $0.50$                            & $R_3$ & $2.0$                              \\
$v_{z,\text{land}}$ & $-0.35$                            &  &                               \\
\bottomrule
\end{tabular}
\end{table}

\subsection{Computational performance and horizon selection} \label{sec:calc_times}
The proposed efficient horizon search from Section \ref{sec:rvhmpc} aims at reducing the computational complexity of robust VH-MPC, while still maintaining its effectiveness.
To validate this, the discretized linear system in \eqref{eq:linearized_system} was simulated in Python with a fixed time step.
In these simulations, the initial UAV positions are $1.0\pm0.5$ [$m$] above the target state, with horizontal deviations of $1.0 \pm 0.5$ [$m$].
Noise is bounded such that $||\bm{w}_{i,\text{max}}||<=1.0$.
In this analytical simulation, the platform target state has a roll angle of $15$ degrees and a pitch of $30$ degrees.
All simulations are performed on a Raspberry Pi 5 to show that the method is practical for a relatively lightweight computer.
For each horizon set dependent on $\Delta N$, 1000 simulation runs are performed.

From the results in Table \ref{tab:performance-vs-all}, it is clear that restricting the horizon search substantially reduces the computation times with only a limited increase in cumulative objective costs.
\begin{table}[b]
\centering
\caption{Effect of $\Delta N$ on landing performance for Python simulations, assuming instant robust VH-MPC calculation times}
\label{tab:performance-vs-all}
\begin{tabular}{lcccc}
\toprule
$\Delta N$ & Solve Time & Landing Time & Landing & Total Cost \\
& (ms) & (s) & Accuracy & $\times 10^3$\\
& & & (m) & \\
\midrule
0 & $5.03 \pm 0.10$ & $1.50 \pm 0.00$ & $0.15 \pm 0.03$ & $14.97 \pm 9.9$ \\ 1 & $15.10 \pm 1.09$ & $1.72 \pm 0.20$ & $0.14 \pm 0.04$ & $14.38 \pm 8.7$ \\ 2 & $23.18 \pm 1.30$ & $1.72 \pm 0.21$ & $0.14 \pm 0.04$ & $14.37 \pm 8.7$ \\ 4 & $38.97 \pm 1.78$ & $1.71 \pm 0.21$ & $0.14 \pm 0.04$ & $14.36 \pm 8.8$ \\ 10 & $80.86 \pm 2.43$ & $1.71 \pm 0.22$ & $0.14 \pm 0.04$ & $14.35 \pm 8.8$ \\ All & $146.9 \pm 2.4$ & $1.71 \pm 0.22$ & $0.14 \pm 0.04$ & $14.35 \pm 8.8$ \\
\bottomrule
\end{tabular}
\end{table}
For example, at $\Delta N = 4$, the controller evaluates nine candidate horizons and remains below a control period of 50 ms for all iterations, with a total cost almost equal to considering all horizons, indicating that the proposed method is computationally suitable for 20 Hz real-time operation.

\subsection{Initial feasibility}
Proposition \ref{propo:recursive_feasibility} guarantees recursive feasibility only when the initial optimization problem is feasible. 
Therefore, we characterize the set of relative UAV positions from which the proposed controller can initiate landing for the parameters in Table \ref{tab:used-parameters}.
When hovering and for a flat platform, the initial position $p_0$ should lie in $\{p: p_x \in [-2.93, 2.93],\, p_y \in [-2.93, 2.93],\, p_z \in [-0.10, 3.47]\}$.
For a moving platform with a roll of 30 degrees, pitch of 15 degrees, and velocity of $v_t=\begin{bmatrix}
    1.0 & 2.0 & 0.5
\end{bmatrix}^\top$, and a moving UAV with a roll of 15 degrees, pitch of 30 degrees, and a horizontal velocity of $v = \begin{bmatrix}
    2.0 & -1.0 & 0.0
\end{bmatrix}^\top$, $p_0$ should lie in $\{p: p_x \in [-3.74, 1.31],\, p_y \in [-2.93, 1.92],\, p_z \in [-0.10, 3.35]\}$.
These results show that, for varying initial conditions, landing can be initiated from a safe distance from the platform.

\subsection{Robustness} \label{sec:robustness}
A Gazebo setup simulates a virtual Holybro X500 quadrotor together with a 6-DOF moving platform, using the architecture in Fig. \ref{fig:control_loop}.
These simulations are performed in real time and include computation delays.
For all simulations, the platform moves periodically with the amplitudes and frequencies given in Table \ref{tab:platform_motion}, together with the range of initial hovering positions from which the landing maneuvers were started.
Throughout the simulation, the ground truth states of the UAV and platform are available for the robust VH-MPC implementation, and the required platform motion prediction is estimated by an Extended Kalman Filter (EKF).
\begin{table}[b] \centering \caption{Initial UAV positions and platform motion parameters} \label{tab:platform_motion} \begin{tabular}{lccc} \toprule & $x$ & $y$ & $z$ \\ \midrule Initial UAV position [m] & $[2.5,3.0]$ & $[-1.0,0.0]$ & $2.5$ \\ Platform amplitude [m] & $0.2$ & $0.3$ & $0.3$ \\ Platform frequency [rad/s] & $0.3$ & $0.5$ & $0.8$ \\ \midrule & $\phi$ & $\theta$ & $\psi$ \\ \midrule Platform amplitude [deg] & $10$ & $20$ & $5$ \\ Platform frequency [rad/s] & $0.4$ & $0.2$ & $0.1$ \\ \bottomrule \end{tabular} \end{table}

The results in Table \ref{tab:gazebo} demonstrate the importance of the efficient horizon search.
For each $\Delta N$, around 25 simulations were performed.
Although evaluating all horizons provides lower cumulative costs, the success rate is substantially lower when real-time computation delays are included.
From this, $\Delta N=2$ provides the best real-time performance in terms of the cost function and success rate, while $\Delta N=0$ provides the highest mission success rate.
\begin{table}[b]
\centering
\caption{Effect of $\Delta N$ on landing performance during real-time Gazebo simulations}
\label{tab:gazebo}
\begin{tabular}{lcccc}
\toprule
$\Delta N$ & Success & Landing Time & Landing & Total Cost \\
           & Rate    & (s)          & Accuracy              &            \\
& & & (m) & \\
\midrule
0 & $92.9\%$ & $2.00 \pm 0.66$ & $0.17 \pm 0.01$ & $96.3 \pm 27.9$ \\ 1 & $70.8\%$ & $2.20 \pm 1.24$ & $0.16 \pm 0.00$ & $110.4 \pm 74.2$ \\ 2 & $72.0\%$ & $1.76 \pm 0.84$ & $0.17 \pm 0.00$ & $88.7 \pm 17.5$ \\ 4 & $67.9\%$ & $2.46 \pm 1.66$ & $0.17 \pm 0.00$ & $99.0 \pm 31.4$ \\ 10 & $50.0\%$ & $3.04 \pm 2.18$ & $0.17 \pm 0.01$ & $116.2 \pm 71.3$ \\ All & $0.0\%$ & --- & --- & --- \\
\bottomrule
\end{tabular}
\end{table}

\subsection{UAV landing experiment}
To validate the effectiveness of the proposed architecture on a real UAV, the experimental setup in Fig. \ref{fig:experimental_setup} was used.
This setup consists of a Holybro X500 quadcopter with a Pixhawk 6C flight controller and an onboard Raspberry Pi 5 for the robust VH-MPC control loop.
To mimic vessel motions, an Acrome Stewart Platform \cite{noauthor_stewart_nodate} has been used to land on.
Since general vessel motions are close to sinusoidal motions \cite{fossen_handbook_2021}, the Stewart platform was controlled to move in sinusoidal motions, specifically, a harmonic pitch motion with an amplitude of 15 degrees at a period of 7 seconds.
An OptiTrack Motion Capture system provides UAV and platform position and attitude measurements.
In the experiments, the reference estimation in the robust VH-MPC problem in \eqref{eq:outer_optimization} is represented by a platform state estimation in 6 DOFs over time for the future horizons, obtained using a simple Extended Kalman Filter (EKF).
The standard PX4 EKF was used to obtain full state measurements of the UAV, incorporating the UAV motion capture measurements.
In an offshore application, these state estimates can be obtained using computer vision measurements, for example.
The predicted state of the moving platform was not directly used withing the robust VH-MPC formulation; instead, a downward landing velocity of $v_{z,\text{land}}$, normal to the platform, was considered to ensure a timely and decisive touchdown.
During the experiments, the platform motion estimates and robust VH-MPC control inputs were calculated in real-time on an onboard Raspberry Pi 5, following the architecture in Fig. \ref{fig:control_loop}.
The robust VH-MPC control scheme was implemented in Python with acados \cite{robin_verschueren_acados_2021}, using an SQP-based formulation with Gauss-Newton Hessian approximation and the partial-condensing HPIPM QP solver, running at 20 Hz.
The calculation of the RPI set and building the MPC formulation in acados can be done offline.
For the state and input bounds $\mathcal{X}$ and $\mathcal{U}$, box constraints were used.
The disturbances during the experiments are assumed to be bounded using a box constraint, where $|\bm{w}_1| \leq \begin{bmatrix}
    0.0 & 0.0 & 0.10
\end{bmatrix}^\top$, $|\bm{w}_2| \leq \begin{bmatrix}
    0.0 & 0.0 & 0.10
\end{bmatrix}^\top$, and $|\bm{w}_3| \leq \begin{bmatrix}
    0.0 & 0.0 & 0.15
\end{bmatrix}^\top$.
The robust VH-MPC formulation assumes bounded disturbances.
However, in implementation, if no feasible solution is found within the nominal search range, the horizon range is extended as a practical fallback.

During the experiments with the setup in Fig. \ref{fig:experimental_setup}, the UAV was manually flown to a feasible state above the moving platform, $0.75-1.4$ [$m$] above the platform, after which the controller from Section \ref{sec:uav_adaptation} was enabled. 

Fig. \ref{fig:3D_trajectory} shows the 3D trajectory of the UAV during one of these experiments, including a projection of the tube set $\mathcal{E}$, and Fig. \ref{fig:state_evolution} shows the evolution of the states during the landing procedure of a typical experiment.
Here, the real state remains within the RPI set around the nominal trajectory throughout the landing procedure, while the nominal trajectory matches the terminal target state, i.e., the moving platform.
The decoupling of the UAV dynamics explains the rectangular shape of the RPI tube.
The deviation of the real state with respect to the nominal state, permitted by the tube, is relatively high for the position, whereas the velocity and tilt deviations are smaller, which is particularly important for a successful landing.
The $y$-position and velocity, as well as the roll of the UAV, are omitted in Fig. \ref{fig:state_evolution} since these are similar to the $x$ state and exhibit less pronounced control behavior, since the roll of the platform was close to zero.
\begin{figure}[t]
    \centering
    \vspace{5pt}
    \includegraphics[width=0.85\linewidth]{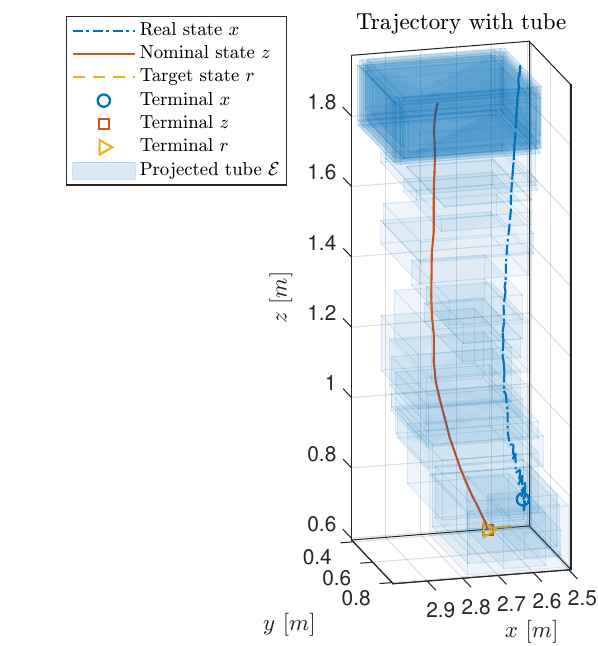}
    \vspace{-5pt}
    \caption{3D UAV landing trajectory and projected RPI tube $\mathcal{E}$. The nominal trajectory reaches the moving target state $r$ at the terminal time, while the measured trajectory remains within the tube.}
    \label{fig:3D_trajectory}
\end{figure}
\begin{figure}[t]
    \centering
    \includegraphics[width=0.95\linewidth]{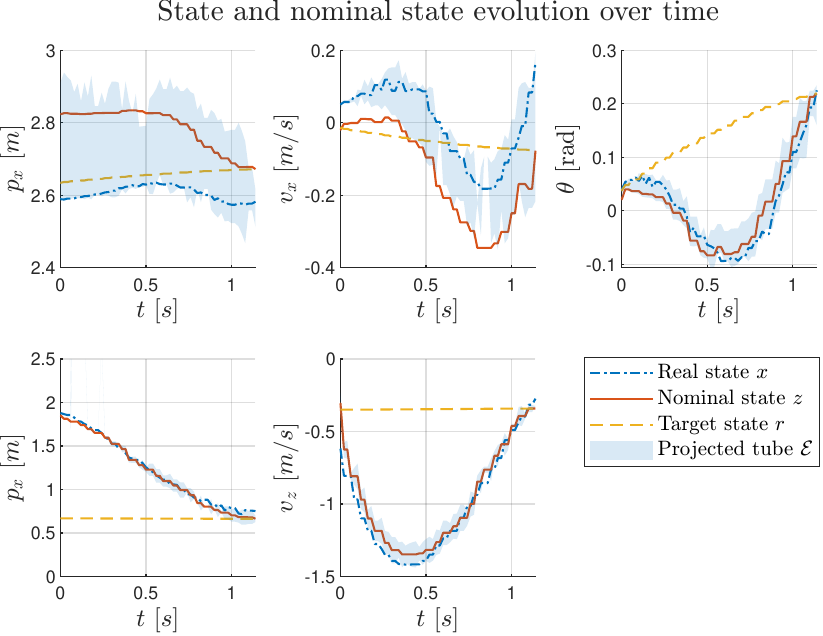}
    \vspace{-5pt}
    \caption{State evolution over time during the landing procedure. The measured state $x$ mostly remains within the projected tube.}
    \vspace{-10pt}
    \label{fig:state_evolution}
\end{figure}

Fig. \ref{fig:cost_per_horizon} highlights the selection of the horizon length from the given range using $\Delta N$.
The selected horizons demonstrate how the horizon can adapt while searching only a small neighborhood around the previously selected horizon.
This, together with the results in Table \ref{tab:gazebo}, indicates that a small horizon search is sufficient for effective real-time horizon adaptation.
\begin{figure}[t]
    \centering
    \vspace{5pt}
    \includegraphics[width=0.99\linewidth]{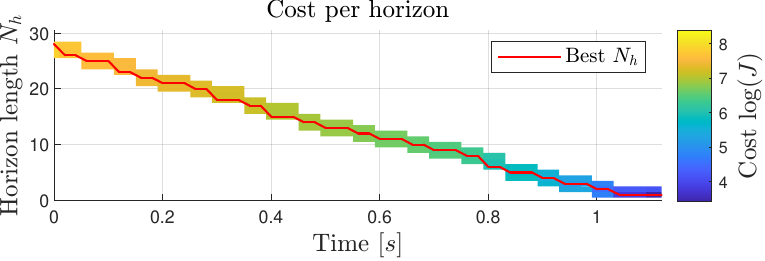}
    \vspace{-10pt}
    \caption{Cost per horizon and selected horizon over time. The cost-to-go is reduced at every time step, while the horizon length has the flexibility to change to achieve a lower cost-to-go.}
    \vspace{-10pt}
    \label{fig:cost_per_horizon}
\end{figure}

\subsection{Discussion}
The use of a fixed RPI tube introduces a trade-off between robustness and landing accuracy.
Larger disturbance bounds result in a larger tube and therefore more conservative tightened constraints, which can limit how closely the UAV can match the platform state at touchdown.

During the experiments, disturbances are never fully bounded within $\mathcal{W}$ and can still lead to infeasible candidate horizons.
The implemented extension of the horizon search provides a practical fallback, but this behavior is not covered by the recursive-feasibility guarantee.

The physical experiments demonstrate the feasibility of the complete onboard architecture and robust control behavior, but the experimental conditions remain simpler than those of a real vessel.
The Stewart platform reproduces representative periodic platform motion, but is limited to a harmonic pitch motion.
Consequently, the proposed architecture has not yet been experimentally evaluated under simultaneous 6-DOF vessel motion.
Such experiments are an important next step toward validating the approach under more representative landing conditions.

\section{CONCLUSION} \label{sec:conclusion}
The robust VH-MPC method allows a system to robustly reach a non-equilibrium point, while the extension to a multirotor UAV system offers a method to robustly land on a 6-DOF moving platform in real-time.
Experiments showcase how the method can be run on a lightweight Raspberry Pi onboard computer to control the body rates and thrust reliably at 20 Hz.
Future work includes using vision-based measurements for platform detection to allow running the algorithm fully onboard the UAV, without the need for external measurements.
Including the measurement and vessel motion modelling errors could then also be incorporated in the robust VH-MPC framework.
When implemented successfully, the proposed method can be tested in a real scenario to land a multirotor UAV robustly on a vessel in the presence of wind disturbances.


\bibliographystyle{Styles/IEEEtran}
\bibliography{Styles/IEEEabrv,refs_no_url}

@article{gonzalez_online_2011,
	title = {Online robust tube-based {MPC} for time-varying systems: a practical approach},
	volume = {84},
	issn = {0020-7179, 1366-5820},
	shorttitle = {Online robust tube-based {MPC} for time-varying systems},
	doi = {10.1080/00207179.2011.594093},
	language = {en},
	number = {6},
	journal = {International Journal of Control},
	author = {Gonzalez, R. and Fiacchini, M. and Alamo, T. and Guzman, J.L. and Rodriguez, F.},
	month = jun,
	year = {2011},
	pages = {1157--1170},
}

@article{saccani_homothetic_2023,
	title = {Homothetic {Tube} {Model} {Predictive} {Control} {With} {Multi}-{Step} {Predictors}},
	volume = {7},
	copyright = {https://ieeexplore.ieee.org/Xplorehelp/downloads/license-information/IEEE.html},
	issn = {2475-1456},
	doi = {10.1109/LCSYS.2023.3336261},
	journal = {IEEE Control Systems Letters},
	author = {Saccani, Danilo and Ferrari-Trecate, Giancarlo and Zeilinger, Melanie N. and Köhler, Johannes},
	year = {2023},
	pages = {3561--3566},
}

@article{zhang_inspection_2024,
	title = {Inspection of {Floating} {Offshore} {Wind} {Turbines} {Using} {Multi}-{Rotor} {Unmanned} {Aerial} {Vehicles}: {Literature} {Review} and {Trends}},
	volume = {24},
	issn = {1424-8220},
	shorttitle = {Inspection of {Floating} {Offshore} {Wind} {Turbines} {Using} {Multi}-{Rotor} {Unmanned} {Aerial} {Vehicles}},
	doi = {10.3390/s24030911},
	language = {en},
	number = {3},
	journal = {Sensors},
	author = {Zhang, Kong and Pakrashi, Vikram and Murphy, Jimmy and Hao, Guangbo},
	month = jan,
	year = {2024},
	pages = {911},
}

@article{richards_robust_2006,
	title = {Robust variable horizon model predictive control for vehicle maneuvering},
	volume = {16},
	copyright = {Copyright © 2006 John Wiley \& Sons, Ltd.},
	issn = {1099-1239},
	doi = {10.1002/rnc.1059},
	language = {en},
	number = {7},
	journal = {International Journal of Robust and Nonlinear Control},
	author = {Richards, Arthur and How, Jonathan P.},
	year = {2006},
	pages = {333--351},
}

@article{robin_verschueren_acados_2021,
	title = {acados -- a modular open-source framework for fast embedded optimal control},
	journal = {Mathematical Programming Computation},
	author = {{Robin Verschueren} and {Gianluca Frison} and {Dimitris Kouzoupis} and {Jonathan Frey} and {Niels van Duijkeren} and {Andrea Zanelli} and {Branimir Novoselnik} and {Thivaharan Albin} and {Rien Quirynen} and {Moritz Diehl}},
	year = {2021},
}

@article{kan_analysis_2019,
	title = {Analysis of {Ground} {Effect} for {Small}-{Scale} {UAVs} in {Forward} {Flight}},
	volume = {4},
	copyright = {https://ieeexplore.ieee.org/Xplorehelp/downloads/license-information/IEEE.html},
	issn = {2377-3766, 2377-3774},
	doi = {10.1109/LRA.2019.2929993},
	number = {4},
	journal = {IEEE Robotics and Automation Letters},
	author = {Kan, Xinyue and Thomas, Justin and Teng, Hanzhe and Tanner, Herbert G. and Kumar, Vijay and Karydis, Konstantinos},
	month = oct,
	year = {2019},
	pages = {3860--3867},
}

@misc{noauthor_stewart_nodate,
	title = {Stewart {Platform}: 6-{DoF} {Hexapod} {Positioner} {System} {\textbar} {Acrome}},
	shorttitle = {Stewart {Platform}},
	language = {en},
}

@article{kothare_robust_1996,
	title = {Robust constrained model predictive control using linear matrix inequalities},
	volume = {32},
	copyright = {https://www.elsevier.com/tdm/userlicense/1.0/},
	issn = {00051098},
	doi = {10.1016/0005-1098(96)00063-5},
	language = {en},
	number = {10},
	journal = {Automatica},
	author = {Kothare, Mayuresh V. and Balakrishnan, Venkataramanan and Morari, Manfred},
	month = oct,
	year = {1996},
	pages = {1361--1379},
}

@book{fossen_handbook_2021,
	edition = {1},
	title = {Handbook of {Marine} {Craft} {Hydrodynamics} and {Motion} {Control}},
	copyright = {http://doi.wiley.com/10.1002/tdm\_license\_1.1},
	isbn = {978-1-119-57505-4 978-1-119-57501-6},
	doi = {10.1002/9781119575016},
	language = {en},
	publisher = {Wiley},
	author = {Fossen, Thor I},
	month = jun,
	year = {2021},
}

@inproceedings{mellinger_minimum_2011,
	address = {Shanghai, China},
	title = {Minimum snap trajectory generation and control for quadrotors},
	isbn = {978-1-61284-386-5},
	doi = {10.1109/ICRA.2011.5980409},
	booktitle = {2011 {IEEE} {International} {Conference} on {Robotics} and {Automation}},
	publisher = {IEEE},
	author = {Mellinger, Daniel and Kumar, Vijay},
	month = may,
	year = {2011},
	pages = {2520--2525},
}

@article{mayne_robust_2005,
	title = {Robust model predictive control of constrained linear systems with bounded disturbances},
	volume = {41},
	copyright = {https://www.elsevier.com/tdm/userlicense/1.0/},
	issn = {00051098},
	doi = {10.1016/j.automatica.2004.08.019},
	language = {en},
	number = {2},
	journal = {Automatica},
	author = {Mayne, D.Q. and Seron, M.M. and Raković, S.V.},
	month = feb,
	year = {2005},
	pages = {219--224},
}

@article{liu_research_2024,
	title = {Research on {Motion} {Control} and {Compensation} of {UAV} {Shipborne} {Autonomous} {Landing} {Platform}},
	volume = {15},
	issn = {2032-6653},
	doi = {10.3390/wevj15090388},
	language = {en},
	number = {9},
	journal = {World Electric Vehicle Journal},
	author = {Liu, Xin and Shao, Mingzhi and Zhang, Tengwen and Zhou, Hansheng and Song, Lei and Jia, Fengguang and Sun, Chengmeng and Yang, Zhuoyi},
	month = aug,
	year = {2024},
	pages = {388},
}

@article{gupta_landing_2023,
	title = {Landing a {UAV} in {Harsh} {Winds} and {Turbulent} {Open} {Waters}},
	volume = {8},
	copyright = {https://ieeexplore.ieee.org/Xplorehelp/downloads/license-information/IEEE.html},
	issn = {2377-3766, 2377-3774},
	doi = {10.1109/LRA.2022.3231831},
	number = {2},
	journal = {IEEE Robotics and Automation Letters},
	author = {Gupta, Parakh M. and Pairet, Eric and Nascimento, Tiago and Saska, Martin},
	month = feb,
	year = {2023},
	pages = {744--751},
}

@inproceedings{iida_adaptive_2025,
	address = {Atlanta, GA, USA},
	title = {Adaptive {Perching} and {Grasping} by {Aerial} {Robot} with {Light}-{Weight} and {High} {Grip}-{Force} {Tendon}-{Driven} {Three}-{Fingered} {Hand} {Using} {Single} {Actuator}},
	copyright = {https://doi.org/10.15223/policy-029},
	isbn = {979-8-3315-4139-2},
	doi = {10.1109/ICRA55743.2025.11128377},
	booktitle = {2025 {IEEE} {International} {Conference} on {Robotics} and {Automation} ({ICRA})},
	publisher = {IEEE},
	author = {Iida, Hisaaki and Sugihara, Junichiro and Sugihara, Kazuki and Kozuka, Haruki and Li, Jinjie and Nagato, Keisuke and Zhao, Moju},
	month = may,
	year = {2025},
	pages = {15218--15224},
}

@inproceedings{lee_autonomous_2024,
	address = {Yokohama, Japan},
	title = {Autonomous aerial perching and unperching using omnidirectional tiltrotor and switching controller},
	copyright = {https://doi.org/10.15223/policy-029},
	isbn = {979-8-3503-8457-4},
	doi = {10.1109/ICRA57147.2024.10610445},
	booktitle = {2024 {IEEE} {International} {Conference} on {Robotics} and {Automation} ({ICRA})},
	publisher = {IEEE},
	author = {Lee, Dongjae and Hwang, Sunwoo and Byun, Jeonghyun and Lee, Seung Jae and Jin Kim, H.},
	month = may,
	year = {2024},
	pages = {1590--1596},
}

@article{sanchez-cuevas_characterization_2017,
	title = {Characterization of the {Aerodynamic} {Ground} {Effect} and {Its} {Influence} in {Multirotor} {Control}},
	volume = {2017},
	copyright = {http://creativecommons.org/licenses/by/4.0/},
	issn = {1687-5966, 1687-5974},
	doi = {10.1155/2017/1823056},
	language = {en},
	journal = {International Journal of Aerospace Engineering},
	author = {Sanchez-Cuevas, Pedro and Heredia, Guillermo and Ollero, Anibal},
	year = {2017},
	pages = {1--17},
}

@article{meere_x-ray_2025,
	title = {X-ray {Image} {Generation} for {Robotic} {Radiography}: a {Case} {Study} on {Motion} {Blur} in {Drone}-{Based} {Wind} {Turbine} {Inspections}},
	volume = {44},
	issn = {0195-9298, 1573-4862},
	shorttitle = {X-ray {Image} {Generation} for {Robotic} {Radiography}},
	doi = {10.1007/s10921-025-01279-6},
	language = {en},
	number = {4},
	journal = {Journal of Nondestructive Evaluation},
	author = {Meere, Bas and Doodeman, Sander and Vidal, Franck P. and Chanfreut, Paula and Torta, Elena and Antunes, Duarte},
	month = dec,
	year = {2025},
	pages = {139},
}

@article{fun_sang_cepeda_exploring_2023,
	title = {Exploring {Autonomous} and {Remotely} {Operated} {Vehicles} in {Offshore} {Structure} {Inspections}},
	volume = {11},
	issn = {2077-1312},
	doi = {10.3390/jmse11112172},
	language = {en},
	number = {11},
	journal = {Journal of Marine Science and Engineering},
	author = {Fun Sang Cepeda, Maricruz and Freitas Machado, Marcos De Souza and Sousa Barbosa, Fabrício Hudson and Santana Souza Moreira, Douglas and Legaz Almansa, Maria José and Lourenço De Souza, Marcelo Igor and Caprace, Jean-David},
	month = nov,
	year = {2023},
	pages = {2172},
}

@article{trodden_one-step_2016,
	title = {A {One}-{Step} {Approach} to {Computing} a {Polytopic} {Robust} {Positively} {Invariant} {Set}},
	volume = {61},
	copyright = {https://ieeexplore.ieee.org/Xplorehelp/downloads/license-information/IEEE.html},
	issn = {0018-9286, 1558-2523},
	doi = {10.1109/TAC.2016.2541300},
	number = {12},
	journal = {IEEE Transactions on Automatic Control},
	author = {Trodden, Paul},
	month = dec,
	year = {2016},
	pages = {4100--4105},
}

@inproceedings{stephenson_distributed_2024,
	address = {Chania - Crete, Greece},
	title = {Distributed {Model} {Predictive} {Control} for {Cooperative} {Multirotor} {Landing} on {Uncrewed} {Surface} {Vessel} in {Waves}},
	copyright = {https://doi.org/10.15223/policy-029},
	isbn = {979-8-3503-5788-2},
	doi = {10.1109/ICUAS60882.2024.10557042},
	booktitle = {2024 {International} {Conference} on {Unmanned} {Aircraft} {Systems} ({ICUAS})},
	publisher = {IEEE},
	author = {Stephenson, Jess and Duncan, Nathan T. and Greeff, Melissa},
	month = jun,
	year = {2024},
	pages = {645--651},
}

@article{prochazka_model_2024,
	title = {Model predictive control-based trajectory generation for agile landing of unmanned aerial vehicle on a moving boat},
	volume = {313},
	issn = {00298018},
	doi = {10.1016/j.oceaneng.2024.119164},
	language = {en},
	journal = {Ocean Engineering},
	author = {Procházka, Ondřej and Novák, Filip and Báča, Tomáš and Gupta, Parakh M. and Pěnička, Robert and Saska, Martin},
	month = dec,
	year = {2024},
	pages = {119164},
}

@inproceedings{vlantis_quadrotor_2015,
	address = {Seattle, WA, USA},
	title = {Quadrotor landing on an inclined platform of a moving ground vehicle},
	isbn = {978-1-4799-6923-4},
	doi = {10.1109/ICRA.2015.7139490},
	booktitle = {2015 {IEEE} {International} {Conference} on {Robotics} and {Automation} ({ICRA})},
	publisher = {IEEE},
	author = {Vlantis, Panagiotis and Marantos, Panos and Bechlioulis, Charalampos P. and Kyriakopoulos, Kostas J.},
	month = may,
	year = {2015},
	pages = {2202--2207},
}

@inproceedings{habas_ceilings_2025,
	address = {Atlanta, GA, USA},
	title = {From {Ceilings} to {Walls}: {Universal} {Dynamic} {Perching} of {Quadrotors} on {Surfaces} with {Variable} {Orientations}},
	copyright = {https://doi.org/10.15223/policy-029},
	isbn = {979-8-3315-4139-2},
	shorttitle = {From {Ceilings} to {Walls}},
	doi = {10.1109/ICRA55743.2025.11128577},
	booktitle = {2025 {IEEE} {International} {Conference} on {Robotics} and {Automation} ({ICRA})},
	publisher = {IEEE},
	author = {Habas, Bryan and Brown, Aaron and Lee, Donghyeon and Goldman, Mitchell and Cheng, Bo},
	month = may,
	year = {2025},
	pages = {288--294},
}

@article{greer_shrinking_2020,
	title = {Shrinking {Horizon} {Model} {Predictive} {Control} {Method} for {Helicopter}–{Ship} {Touchdown}},
	volume = {43},
	issn = {0731-5090},
	doi = {10.2514/1.G004374},
	number = {5},
	journal = {Journal of Guidance, Control, and Dynamics},
	publisher = {American Institute of Aeronautics and Astronautics},
	author = {Greer, William B. and Sultan, Cornel},
	month = may,
	year = {2020},
	pages = {884--900},
}

@article{rodriguez-ramos_deep_2019,
	title = {A {Deep} {Reinforcement} {Learning} {Strategy} for {UAV} {Autonomous} {Landing} on a {Moving} {Platform}},
	volume = {93},
	issn = {1573-0409},
	doi = {10.1007/s10846-018-0891-8},
	language = {en},
	number = {1},
	journal = {Journal of Intelligent \& Robotic Systems},
	author = {Rodriguez-Ramos, Alejandro and Sampedro, Carlos and Bavle, Hriday and de la Puente, Paloma and Campoy, Pascual},
	month = feb,
	year = {2019},
	pages = {351--366},
}


\addtolength{\textheight}{-12cm}   

\end{document}